\PassOptionsToPackage{unicode}{hyperref}
\PassOptionsToPackage{hyphens}{url}
\documentclass[
  12pt,
  a4paper,
]{article}
\usepackage{amsmath,amssymb}
\usepackage{amsthm}
\newtheorem{theorem}{Theorem}
\newtheorem{lemma}{Lemma}
\usepackage[a4paper,margin=27mm]{geometry}
\usepackage{iftex}
\ifPDFTeX
  \usepackage[T1]{fontenc}
  \usepackage[utf8]{inputenc}
  \usepackage{textcomp} % provide euro and other symbols
\else % if luatex or xetex
  \usepackage{unicode-math} % this also loads fontspec
  \defaultfontfeatures{Scale=MatchLowercase}
  \defaultfontfeatures[\rmfamily]{Ligatures=TeX,Scale=1}
\fi
\usepackage{lmodern}
\ifPDFTeX\else
\fi
\IfFileExists{upquote.sty}{\usepackage{upquote}}{}
\IfFileExists{microtype.sty}{% use microtype if available
  \usepackage[]{microtype}
  \UseMicrotypeSet[protrusion]{basicmath} % disable protrusion for tt fonts
}{}
\makeatletter
\@ifundefined{KOMAClassName}{% if non-KOMA class
  \IfFileExists{parskip.sty}{%
    \usepackage{parskip}
  }{% else
    \setlength{\parindent}{0pt}
    \setlength{\parskip}{6pt plus 2pt minus 1pt}}
}{% if KOMA class
  \KOMAoptions{parskip=half}}
\makeatother
\usepackage{xcolor}
\usepackage{color}
\usepackage{fancyvrb}

\DefineVerbatimEnvironment{Highlighting}{Verbatim}{commandchars=\\\{\}}
\newenvironment{Shaded}{}{}

\newcommand{\ExtensionTok}[1]{#1}

\newcommand{\NormalTok}[1]{#1}

\newcommand{\cliopt}[1]{\texttt{-{}-#1}}
\usepackage{longtable,booktabs,array}
\usepackage{calc} % for calculating minipage widths
\usepackage{etoolbox}
\makeatletter
\patchcmd\longtable{\par}{\if@noskipsec\mbox{}\fi\par}{}{}
\makeatother
\IfFileExists{footnotehyper.sty}{\usepackage{footnotehyper}}{\usepackage{footnote}}
\makesavenoteenv{longtable}
\providecommand{\tightlist}{%
  \setlength{\itemsep}{0pt}\setlength{\parskip}{0pt}}
\ifLuaTeX
  \usepackage{selnolig}  % disable illegal ligatures
\fi
\usepackage[numbers,sort&compress]{natbib}
\usepackage{bookmark}
\IfFileExists{xurl.sty}{\usepackage{xurl}}{} % add URL line breaks if available
\hypersetup{
  pdftitle={Dependency-Aware ROM/CBD Correctness Bounds for ML-KEM-768 at the Heuristic Failure Scale},
  hidelinks,
  pdfcreator={LaTeX via pandoc}}

\title{Dependency-Aware ROM/CBD Correctness Bounds for ML-KEM-768 at the
Heuristic Failure Scale}
\author{%
Aur\'elie Duriez\textsuperscript{1} and Christophe Tommasini\textsuperscript{*2}\\[0.75em]
\small \textsuperscript{1}netHsys SARL, Lille, France\\
\small \textsuperscript{2}Tommasini Conseil, Lille, France%
}
\date{}

\begin{document}
\maketitle
\begingroup
\renewcommand{\thefootnote}{\fnsymbol{footnote}}
\footnotetext[1]{Corresponding author: \texttt{contact@tommasini-conseil.com}}
\endgroup
\begin{abstract}
We certify an honest-decapsulation failure upper bound for ML-KEM-768 in an explicit random-function/centered-binomial (ROM/CBD) abstraction. Domain-separated public-matrix streams are modeled as independent uniform ring elements and secret/noise polynomials as independent CBD2 primitives; this is not an information-theoretic statement about the fixed SHAKE instantiation of FIPS 203. Recent formal assessments identify rigorous justification of ML-KEM's heuristic decapsulation-failure scale as an open problem; within the explicit ROM/CBD abstraction studied here, we obtain a dependency-preserving certified upper bound at that scale. The analysis preserves dependencies induced by the public matrix and by both ciphertext-compression terms. Its terminal chain has three components: an exact graph-coupled full-ideal reference for the joint $c_u/c_v$ residual; a proper-ideal bivariate Fourier transport whose rare $|T|\ge3$ branch is closed by an exhaustive three-factor anti-concentration replay; and exact bit-specific FIPS decoding events followed only by a 256-coordinate union bound. A formal partial-Fourier lemma makes the spectral-to-total-variation step explicit. The reduced rational certificate satisfies
\[
  \Pr[K'\neq K] \le P_* \le 2^{-164.81},
\]
with $-\log_2 P_*=164.810716201343121\ldots$. The $164.81$ threshold is exact but numerically tight: the certified exponent exceeds it by only about $0.0007162$ bit, and $164.82$ is not certified. The result is an upper bound for an arbitrary message fixed independently of the public and secret randomness, under honest encryption and decapsulation. It is not an exact DFR, not a fixed-SHAKE equivalence theorem, not a new IND-CCA reduction, and not an adaptive $\delta$-correctness result.
\end{abstract}

\section{Introduction}\label{introduction}

ML-KEM, standardized in FIPS 203, inherits from Kyber a very small
probability that an honestly generated ciphertext decrypts to a message
different from the encapsulated one. FIPS 203 reports a failure
probability at approximately the \(2^{-164.8}\) scale for ML-KEM-768
\cite{fips203}. That reported scale is useful context, but it should
not be confused with the theorem proved here: our result is a certified
upper bound in an explicitly stated random-function/CBD abstraction. The
quantitative difficulty is that the compression errors in the two
ciphertext components are generated from algebraically related
quantities and are not naturally independent of the lattice noise
entering the correctness residual.

This distinction has become increasingly explicit in the literature.
Dependency effects in Ring/Module-LWE-type decryption failures were
studied well before ML-KEM standardization
\cite{danvers2019, fang2022}. Almeida et al.~give a machine-checked
correctness and IND-CCA treatment of ML-KEM in EasyCrypt
\cite{almeida2024}, but that formal security framework should not be
confused with a sharp dependency-preserving evaluation of the concrete
failure probability. Barbosa, Kannwischer, Lim, Schwabe, and Strub
separate a provable correctness route from the familiar heuristic
rounding distributions; for ML-KEM-768 their provable route is around
\(2^{-80}\), whereas the \(2^{-164}\)-scale quantity belongs to the
simplified rounding model \cite{barbosa2025}. Most closely related to
our algebraic route, Qayyum and Bezzateev develop a certified
ideal-stratification/Fourier anti-concentration framework with
independently checkable certificates \cite{qayyum2026}. Inspection of
the released version-1.0.0 artifact sharpens the comparison. Its
hardened full-parameter theorem handles the actual \(c_v\) term without
an independence assumption, but only through its deterministic support
envelope; for ML-KEM-768 it uses \(\lvert c_v\rvert\le104\) and
certifies a dependency-preserving total upper bound with reported
\(\log_2\) value \(-127.4217902151193\). The same artifact separately
reproduces the familiar independent-rounding heuristic value near
\(-164.8116822526\) and explicitly labels that computation heuristic.
It also disclaims an exact or true ML-KEM decapsulation-failure
probability.

The purpose of this work is not to claim priority for ideal/coset
stratification, CRT methods, finite Fourier analysis,
anti-concentration, or independently checkable certificates. Those
ingredients are prior art. Our contribution is instead a quantitative
continuation of that framework: we retain an explicit full-ideal failure baseline while controlling
the proper-ideal correction, obtain an exact full-ideal tail for the
\(c_u\) contribution, and then resolve the actual joint \(c_u/c_v\)
dependence through a graph-coupled bivariate group-algebra transport,
rather than replacing \(c_v\) by its worst-case support envelope. This
produces a correctness upper bound at essentially the heuristic failure
scale.

\subsection{Scope}\label{scope}

The theorem is explicitly a theorem in a \textbf{ROM/CBD abstraction}.
Here ``ROM/CBD'' is project shorthand for two probability-model
primitives used only for the correctness analysis: (i) the distinct
domain-separated public-matrix generation streams are treated as
independent uniform ring elements, and (ii) the secret/noise polynomials
are treated as independent centered-binomial samples with independent
CBD2 coefficients. This is not a new random-oracle security reduction,
and it does \textbf{not} assert that the fixed SHAKE instantiation of
FIPS 203 is information-theoretically identical to the abstraction.
Compression, decompression, message encoding, and decoding remain the
exact audited integer maps of the ML-KEM-768 parameter set.

The probability experiment used below fixes an arbitrary 256-bit
message, samples the public matrix and all secret/noise variables
according to these ROM/CBD primitives, performs honest encryption and
decapsulation with the exact audited rounding maps, and asks whether the
returned key differs from the encapsulated key. The theorem is uniform
in the message and therefore applies in particular to the random message
generated during honest encapsulation. Within this scope, the final
result is an upper bound, not an exact DFR. The final step uses only a
union bound over 256 decoded coefficients; no independence between those
coefficients is assumed.

\begin{theorem}[ROM/CBD honest-correctness certificate]\label{thm:main}
Let $q=3329$, $n=256$, $k=3$, $\eta_1=\eta_2=2$, $d_u=10$, and $d_v=4$. Consider the ROM/CBD experiment described above, using the exact FIPS 203 compression, decompression, message-encoding, and message-decoding maps. Fix an arbitrary 256-bit message independently of the public matrix and all secret/noise draws, and perform honest ML-KEM-768 encryption and decapsulation. Then the exact rational certificate $P_*$ satisfies
\[
  \Pr[K'\neq K] \le P_* \le 2^{-164.81}.
\]
This is an upper bound in the stated abstraction. It is not an information-theoretic identity for fixed SHAKE, not an exact DFR, not a new IND-CCA reduction, and does not establish the adaptive correctness notion in which a message may be selected after observing public data.
\end{theorem}

\subsection{Contributions}\label{contributions}

The terminal proof has three publication-level contributions.

\begin{enumerate}
\item \textbf{Graph-coupled full-ideal reference for joint $c_u/c_v$.} We derive the exact bivariate primitive induced by the shared matrix-dependent quantity, prove the partial-Fourier-to-total-variation lemma used for 768 independent primitive contributions, and certify the two critical spectral modes by $853/1000$ and every other nonzero mode by $3/5$.
\item \textbf{Proper-ideal bivariate transport with certified rare-stratum control.} We transport the full-ideal reference through the affine ideal mixture, recompute the $|T|=1$ and $|T|=2$ corrections exactly/safely, and close $|T|\ge3$ with the three-factor Construction-A certificate and exhaustive Fincke--Pohst proof replay. The resulting global affine correction has exponent $169.8071173061\ldots$ bits.
\item \textbf{Bit-specific terminal certificate at the heuristic scale.} We use the exact FIPS bit-0 and bit-1 safe regions, order only their certified graph bounds, and apply a union bound over 256 coefficients. Primary exact reconstruction and an equal-strength secondary C11C assembly with independently derived theorem-dominant upstream inputs both certify $P_*\le2^{-164.81}$; $2^{-164.82}$ is not certified.
\end{enumerate}

Two auxiliary computations are deliberately separated from that dependency chain. The exact full-ideal $c_u$ tail is retained as an intermediate benchmark and diagnostic, but is not a term of the terminal formula. The retired low-weight census/closure development material is not part of the V8 publication capsule, is not a publication claim, and is not an input to Theorem~\ref{thm:main}.

The terminal fraction is checked by exact integer comparisons; decimal logarithms are presentation only.

\section{Related work and
positioning}\label{related-work-and-positioning}

\subsection{Dependency effects and concrete Kyber failure
analyses}\label{dependency-effects-and-concrete-kyber-failure-analyses}

The fact that decryption failures in Ring/Module-LWE/LWR constructions
can exhibit nontrivial dependence is established prior art
\cite{danvers2019}. Fang, Wang, and Zhao develop a tight analysis for
concrete Kyber public keys without imposing a blanket independence
assumption on the underlying error events \cite{fang2022}. Their
distribution-over-keys conclusions, however, rely on a sampled set of
matrices and density estimation; this is conceptually different from a
uniform/certified ROM average bound of the type pursued here.

Accordingly, we do not claim to be the first dependency-aware
Kyber/ML-KEM analysis.

\subsection{Machine-checked correctness and cryptographic
reductions}\label{machine-checked-correctness-and-cryptographic-reductions}

Almeida et al.~provide machine-checked correctness and IND-CCA security
of ML-KEM in EasyCrypt \cite{almeida2024}. Our work does not replace or
strengthen their IND-CCA theorem. We use only the deterministic
correctness implication that, for an honestly generated ciphertext,
recovery of the same K-PKE message causes decapsulation to rederive the
same coins, reconstruct the same ciphertext, and return the same shared
key. Our contribution is the concrete probability bound for
message-recovery failure in the stated stochastic abstraction.

\subsection{Provable versus heuristic correctness
bounds}\label{provable-versus-heuristic-correctness-bounds}

Barbosa et al.~formalize correctness bounds for lattice-based
cryptography and explicitly distinguish the heuristic ML-KEM rounding
distributions from their provable cryptographic correctness route
\cite{barbosa2025}. For ML-KEM-768, the familiar \(2^{-164}\)-scale
value belongs to the simplified model in which compression noise is
treated as arising from uniform inputs, whereas their provable route is
around \(2^{-80}\). Their analysis is the direct motivation for asking
whether the heuristic-scale behavior can be recovered without silently
replacing the actual dependency structure by independent rounding noise.

Barbosa et al. and the CRYPTREC technical assessment explicitly identify formal justification of the heuristic ML-KEM failure scale as an open problem \cite{barbosa2025,cryptrec2026}. The result here addresses a restricted form of that gap: it gives a dependency-preserving upper bound at the heuristic scale only in the fixed-message ROM/CBD abstraction of Theorem~\ref{thm:main}. It does not establish an information-theoretic fixed-SHAKE DFR theorem.

A recent probabilistic study by Yavas, Chen, and Kadlec computes
exact/FFT-based tails for centered-binomial aggregate-noise models and
compares them with classical concentration inequalities
\cite{yavas2026}. We regard that line as complementary: its purpose is
probabilistic tail modeling under simplified aggregate-noise
assumptions, whereas the present work is devoted to preserving the
algebraic dependence induced by the public matrix and by both
compression terms. CRYPTREC's 2026 technical assessment likewise records the heuristic/provable gap and states that formally proving the heuristic ML-KEM failure rate is considered an open problem \cite{cryptrec2026}. This reinforces the motivation for the restricted ROM/CBD result here without turning it into a fixed-SHAKE theorem.

Bajri\'c gives an exact Walsh--Hadamard spectral analysis of ML-KEM compression maps that depends explicitly on the chosen lift/Boolean representation \cite{bajric2026}. That representation-aware compression-map analysis is complementary rather than a collision with the present result: our target is a dependency-preserving correctness/failure upper bound in the ROM/CBD experiment, with a joint bivariate residual and ideal transport.

\subsection{Qayyum--Bezzateev and the shared ideal/Fourier
framework}\label{qayyumbezzateev-and-the-shared-idealfourier-framework}

Qayyum and Bezzateev's 2026 artifact is the closest methodological
neighbor \cite{qayyum2026}. It contains the same broad structural
ingredients that underlie the first half of our proof: the conditional
coset law for \(A^T r\), CRT ideal stratification, finite Fourier
analysis, anti-concentration via character energy, and independently
checkable computation. These ideas are therefore not claimed here as
new.

The released version-1.0.0 artifact makes the quantitative boundary
clear. Its revised full-parameter theorem uses ideal strata \(H_0\),
\(H_1\), and \(H_{\ge2}\), with Fourier anti-concentration on the
multi-factor branch and independently checked character-energy minima.
For the actual \(c_v\) contribution, however, it uses a deterministic
upper-bound envelope rather than a joint full-dimensional \(c_u/c_v\)
law: for ML-KEM-768 the certified margin is reduced by
\(\max\lvert c_v\rvert=104\) to 727. Its independent verifier reports
a dependency-preserving total \(\log_2\) upper bound of
\(-127.4217902151193\), certifying the stated \(2^{-127}\) threshold.
The artifact's separate reproduction near \(2^{-164.8116822526}\) is
explicitly an independent-rounding heuristic reproduction, not that
certified dependency-preserving theorem. Thus the ideal/coset
stratification, Fourier anti-concentration, orbit-energy reduction, and
independently checkable certificate architecture are prior art, whereas
the graph-coupled full-dimensional \(c_u/c_v\) closure at the heuristic
scale is not supplied by that artifact.

\subsection{What we claim, and what we do
not}\label{what-we-claim-and-what-we-do-not}

The defensible novelty claim is therefore narrow and quantitative: in a
stated ROM/CBD abstraction, we carry a dependency-preserving
ideal/Fourier analysis through exact \(c_u\) tails and graph-coupled
\(c_u/c_v\) transport to an upper bound below \(2^{-164.81}\). We do not
claim first use of ideals, CRT stratification, Fourier
anti-concentration, or independently auditable certificates.

\textbf{Table 1 --- Status-aware comparison.} Numerical exponents are
meaningful only together with their model/status.

\begin{longtable}[]{@{}
  >{\raggedright\arraybackslash}p{(\columnwidth - 6\tabcolsep) * \real{0.2500}}
  >{\raggedright\arraybackslash}p{(\columnwidth - 6\tabcolsep) * \real{0.2500}}
  >{\raggedright\arraybackslash}p{(\columnwidth - 6\tabcolsep) * \real{0.2500}}
  >{\raggedright\arraybackslash}p{(\columnwidth - 6\tabcolsep) * \real{0.2500}}@{}}
\toprule\noalign{}
\begin{minipage}[b]{\linewidth}\raggedright
Work
\end{minipage} & \begin{minipage}[b]{\linewidth}\raggedright
Model/status
\end{minipage} & \begin{minipage}[b]{\linewidth}\raggedright
Key distinction for this paper
\end{minipage} & \begin{minipage}[b]{\linewidth}\raggedright
ML-KEM-768 status
\end{minipage} \\
\midrule\noalign{}
\endhead
\bottomrule\noalign{}
\endlastfoot
FIPS 203 \cite{fips203} & Standard specification; reported
parameter-set scale & Reports the familiar scale but is not the
dependency-preserving proof developed here & about \(2^{-164.8}\)
reported scale \\
D'Anvers--Vercauteren--Verbauwhede \cite{danvers2019} & Dependency
analysis for Ring/Mod-LWE/LWR & Establishes that failure dependencies
matter & no directly comparable certified ML-KEM-768 exponent used
here \\
Fang--Wang--Zhao \cite{fang2022} & Concrete Kyber key analysis; sampled
distribution-over-keys step & Dependency-aware for fixed public keys &
not a uniform certified ROM-average bound \\
Almeida et al. \cite{almeida2024} & Machine-checked correctness and
IND-CCA framework & Formal cryptographic reductions, not a sharp
dependency-preserving DFR exponent & no comparable sharp exponent
claimed here \\
Barbosa et al. \cite{barbosa2025} & Formally verified heuristic and
provable routes & Explicitly separates simplified rounding from the
provable route & about \(2^{-80}\) provable route; \(2^{-164}\)-scale
heuristic model \\
Qayyum--Bezzateev \cite{qayyum2026} & Public certified ideal/Fourier artifact
in a related abstraction & Shared ideal/Fourier/anti-concentration foundation;
actual \(c_v\) handled by the deterministic \(\lvert c_v\rvert\le104\)
envelope, not by joint graph transport & dependency-preserving total
\(\log_2\) upper bound \(-127.42179\ldots\); heuristic reproduction
separately near \(-164.81168\) \\
This work & Explicit ROM/CBD correctness abstraction & Failure-aware
proper-ideal transport plus graph-coupled bivariate \(c_u/c_v\)
transport & \(P_*\le2^{-164.81}\) upper bound \\
\end{longtable}

\section{Model, notation, and correctness
event}\label{model-notation-and-correctness-event}

\subsection{Ring and parameters}\label{ring-and-parameters}

Let

\[
  q=3329,\qquad n=256,\qquad k=3,
\]

and

\[
  R_q = \mathbb F_q[X]/(X^{256}+1).
\]

For ML-KEM-768,

\[
  \eta_1=\eta_2=2,\qquad d_u=10,\qquad d_v=4.
\]

A CBD2 coefficient \(C\) takes values in \(\{-2,-1,0,1,2\}\) with counts

\[
  (1,4,6,4,1)/16.
\]

The vectors \(s,e,r,e_1\in R_q^3\) and \(e_2\in R_q\) are independent
samples in the ROM/CBD abstraction, with independent CBD2 coefficients.
The symbol \(r\) denotes the ephemeral vector called \(y\) in FIPS 203.

\textbf{Notation used throughout the proof.}

\begin{longtable}[]{@{}
  >{\raggedright\arraybackslash}p{(\columnwidth - 2\tabcolsep) * \real{0.5000}}
  >{\raggedright\arraybackslash}p{(\columnwidth - 2\tabcolsep) * \real{0.5000}}@{}}
\toprule\noalign{}
\begin{minipage}[b]{\linewidth}\raggedright
Symbol
\end{minipage} & \begin{minipage}[b]{\linewidth}\raggedright
Meaning
\end{minipage} \\
\midrule\noalign{}
\endhead
\bottomrule\noalign{}
\endlastfoot
\(A\) & public \(3\times3\) matrix over \(R_q\) \\
\(r\) & ephemeral vector (called \(y\) in FIPS 203) \\
\(I(r)\) & ideal generated by the three components of \(r\) \\
\(S(r),T(r)\) & active and inactive CRT-factor sets of \(I(r)\) \\
\(J=I(r)I(s)\) & product ideal governing the matrix-dependent part of
\(v\) \\
\(R_d\) & decompress-after-compress map at \(d\) bits \\
\(c_u,c_v\) & coefficientwise compression errors for \(d_u=10,d_v=4\) \\
\(n'\) & correctness residual before adding \(c_v\) \\
\(z=s^T R_{10}(u)\) & rounded secret-product term coupled to \(v\) \\
\(Q,Z\) & full-ideal scalar bivariate primitive coordinates \\
\(M=Q-Z\) & first coordinate used in the bivariate affine transport \\
\(P_*\) & final certified global upper bound \\
\end{longtable}

\subsection{ROM/CBD abstraction}\label{romcbd-abstraction}

The public matrix \(A\in R_q^{3\times3}\) is modeled as independent
uniform ring elements on the distinct domain-separated random-function
streams used to generate the matrix. This idealization is an assumption
of the theorem. It is not an information-theoretic theorem about the
fixed SHAKE implementation, nor does it analyze possible equalities or
collisions among finite seeds/inputs in that literal instantiation.

The model also treats the CBD samples as independent primitives. In
particular, this independence is part of the probability model rather
than a theorem deduced from the finite-seed SHAKE/PRF implementation. We
condition on well-formed inputs and exclude randomness-generation
failure, consistently with the honest-correctness experiment being
bounded.

\subsection{Compression maps}\label{compression-maps}

For \(d\in\{4,10\}\) define the exact FIPS integer maps

\[
  \operatorname{Compress}_q(x,d)
   = \left\lfloor \frac{2^d x}{q}\right\rceil \bmod 2^d,
\]

and

\[
  \operatorname{Decompress}_q(y,d)
   = \left\lfloor \frac{q y}{2^d}\right\rceil \bmod q.
\]

We write

\[
  R_d(x)=\operatorname{Decompress}_q(\operatorname{Compress}_q(x,d),d),
\]

and the centered compression error

\[
  c_d(x)=R_d(x)-x\pmod q.
\]

Thus \(c_u=c_{10}(u)\) and \(c_v=c_4(v)\) coefficientwise.

The displayed nearest-integer expressions are mathematical shorthand for
the exact audited FIPS integer maps. All boundary statements in this
paper, including the bit-specific safe intervals below, are obtained by
exhaustive evaluation of those integer maps; no theorem relies on an
unspecified real-number tie-breaking convention.

\subsection{K-PKE algebra}\label{k-pke-algebra}

The audited K-PKE equations are

\[
  t=As+e,
\]

\[
  u=A^T r+e_1,
\]

and

\[
  v=t^T r+e_2+\mu,
\]

where \(\mu\) is the message polynomial encoding.

For message bit \(b\in\{0,1\}\), write \(\mu_b\in\{0,1665\}\) for the
coefficient encoding and \(c_{v,b}=c_4(v_b)\), with
\(v_b=t^Tr+e_2+\mu_b\). After decompression, the centered
coefficientwise correctness residual is therefore bit-dependent:

\[
  \widetilde n_b
    =\langle e,r\rangle
     -\langle s,e_1\rangle
     -\langle s,c_u\rangle
     +e_2+c_{v,b}.
  \tag{1}
\]

The sign of \(c_{v,b}\) in (1) is positive. When the bit index is
irrelevant to an intermediate algebraic identity, we suppress it in the
notation; the terminal probability comparison below does not do so.

\subsection{Exact bit-specific decoding
regions}\label{exact-bit-specific-decoding-regions}

Exhaustive evaluation of the exact integer message-compression maps
gives the centered safe intervals

\[
  \mathcal S_0=[-832,832]
\]

for bit 0 and

\[
  \mathcal S_1=[-832,831]
\]

for bit 1. In the numerical residual-value space, the corresponding
failure sets are

\[
  \mathcal F_0=(-\infty,-833]\cup[833,\infty),
\]

\[
  \mathcal F_1=(-\infty,-833]\cup[832,\infty),
  \tag{2}
\]

and hence \(\mathcal F_0\subset\mathcal F_1\) as sets of numerical
values. This inclusion does \emph{not} order the true bit-specific
probabilities, because the evaluated residual itself depends on the bit:
\(\widetilde n_0\) and \(\widetilde n_1\) are different functions of the
same underlying sample. In particular, the exact regression witnesses
\((n',z)=(-900,-2117)\) and \((-900,-2705)\) exhibit opposite pointwise
failure orderings for the two bits. The executable regression is shipped
as \texttt{anc/code/bit\_specific\_regression.py}.

\subsection{From K-PKE message recovery to ML-KEM
decapsulation}\label{from-k-pke-message-recovery-to-ml-kem-decapsulation}

\textbf{Proposition 1 (deterministic correctness containment).} Consider
a ciphertext produced by honest ML-KEM encapsulation. If K-PKE.Decrypt
returns the same message \(m\) used by Encaps, then ML-KEM.Decaps
returns the same shared key as Encaps.

\textbf{Proof.} Encaps deterministically derives the candidate shared
key and K-PKE encryption randomness from the encapsulated message
together with the public-key hash. If Decaps recovers the same message,
it evaluates the same derivation inputs and therefore obtains the same
candidate key and the same encryption randomness. Re-encryption then
reconstructs the same honestly generated ciphertext, so the ciphertext
check succeeds and the implicit-rejection branch is not selected. Hence
the returned key equals the encapsulated key. \(\square\)

Therefore

\[
  \Pr[K'\neq K]
  \le
  \Pr[\text{K-PKE message recovery fails}].
  \tag{3}
\]

This is a deterministic event containment. It is not a new
Fujisaki--Okamoto or IND-CCA reduction.

\section{Ideal geometry}\label{ideal-geometry}

This section records only the geometry needed by the later Fourier argument.

\subsection{Conditional coset law}\label{conditional-coset-law}

For fixed \(r=(r_1,r_2,r_3)\in R_q^3\), define

\[
  I(r)=R_qr_1+R_qr_2+R_qr_3.
\]

The \(R_q\)-linear map

\[
  \Phi_r:R_q^3\to R_q,\qquad
  (a_1,a_2,a_3)\mapsto \sum_j a_jr_j
\]

has image exactly \(I(r)\). Since every fiber of a linear map between
finite additive groups has the same cardinality, a uniform matrix column
maps to a uniform element of \(I(r)\). Therefore, conditionally on
\((r,e_1)\),

\[
  u_i=e_{1,i}+W_i,
  \qquad W_i\sim\operatorname{Unif}(I(r)),
  \tag{4}
\]

and the three module components are conditionally independent because
they use independent matrix columns.

This ideal/coset structure is shared with the recent Qayyum--Bezzateev
framework and is not claimed as new.

\subsection{CRT decomposition}\label{crt-decomposition}

Let \(\zeta=17\in\mathbb F_q\) have order 256 and set

\[
  \alpha_\ell=\zeta^{2\ell+1},\qquad 0\le \ell<128.
\]

Then

\[
  X^{256}+1
  =\prod_{\ell=0}^{127}(X^2-\alpha_\ell),
\]

where every quadratic factor is irreducible over \(\mathbb F_q\). Hence

\[
  R_q\cong\prod_{\ell=0}^{127}K_\ell,
  \qquad K_\ell\cong\mathbb F_{q^2}.
  \tag{5}
\]

Each ideal is specified by an active subset of the 128 factors. We write
\(S(r)\) for the active set of \(I(r)\) and \(T(r)=S(r)^c\) for the
inactive set.

A useful marginal fact is that every nonzero ideal projects surjectively
onto every standard coefficient. Thus if \(r\neq0\), each individual
coefficient of \(A^T r\) is uniform in \(\mathbb F_q\). This marginal
uniformity does \textbf{not} imply joint independence of coefficients,
which is precisely why the ideal geometry matters.

\subsection{Auxiliary low-weight material from V2 (not a publication claim)}\label{auxiliary-low-weight-v2}

An earlier development branch contained an exhaustive low-weight census and CBD2 accessibility closures for a separate structural statement. Direct inspection of the terminal dependency graph shows that those outputs are not consumed by C10C, C11B, C11C, C12B, or C13A. We therefore remove that standalone statement and its historical computation from the publication capsule rather than presenting an unused result as a theorem dependency. No sampling result is substituted for the retired claim.

\section{Failure-aware Fourier
transport}\label{failure-aware-fourier-transport}

To make the logical status of the proof explicit, we use the following
categories throughout Sections 5--9.

\begin{longtable}[]{@{}
  >{\raggedright\arraybackslash}p{(\columnwidth - 4\tabcolsep) * \real{0.3333}}
  >{\raggedright\arraybackslash}p{(\columnwidth - 4\tabcolsep) * \real{0.3333}}
  >{\raggedright\arraybackslash}p{(\columnwidth - 4\tabcolsep) * \real{0.3333}}@{}}
\toprule\noalign{}
\begin{minipage}[b]{\linewidth}\raggedright
Status
\end{minipage} & \begin{minipage}[b]{\linewidth}\raggedright
Meaning
\end{minipage} & \begin{minipage}[b]{\linewidth}\raggedright
Representative example
\end{minipage} \\
\midrule\noalign{}
\endhead
\bottomrule\noalign{}
\endlastfoot
\textbf{Identity} & exact algebra in a finite ring/group & affine
Fourier identity (11), graph identity (28) \\
\textbf{Exhaustive certificate} & finite computation covering all
canonical cases & four-coordinate census, three-factor orbit/SVP
checks \\
\textbf{Rational inequality} & one-sided exact/certified numerical bound
& (16), (21), (34), terminal integer comparisons \\
\textbf{Norm/TV domination} & event-universal one-sided transport bound
& proper-ideal affine correction \\
\textbf{Union bound} & no independence assumption required & CRT-stratum
unions and the final 256-coordinate step \\
\end{longtable}

The terminal theorem is a composition of these one-sided statements; it
is never identified with an exact distributional equality.

\subsection{Pulling rounding back to the
ideal}\label{pulling-rounding-back-to-the-ideal}

Fix a nonzero \(r\) and set \(H=I(r)\). For one coefficient \(e\) of
\(e_1\) and one ideal coordinate \(w\), define

\[
  g_e(w)=e+c_u(e+w)=R_{10}(e+w)-w.
  \tag{9}
\]

The identity in (9) is exact for all \(5q\) input pairs. After averaging
a CBD2 secret coefficient \(s\), define

\[
  \psi_t(Z)=\mathbb E_s[Z^{st}],
\]

and

\[
  f_Z(w)=\mathbb E_{e,s}[Z^{s g_e(w)}].
\]

Negacyclic signs disappear under the symmetry of CBD2. For one
polynomial component, the conditional observable becomes a product over
256 standard coefficients,

\[
  \Phi_Z(w)=\prod_{b=0}^{255} f_Z(w_b).
  \tag{10}
\]

All dependence is thus concentrated in the fact that \(w\) is uniform on
\(H\), not on the full vector space.

\subsection{Affine Fourier identity}\label{affine-fourier-identity}

Let \(G=\mathbb F_q^{256}\) with normalized additive Fourier transform

\[
  \widehat\Phi(\xi)
    =|G|^{-1}\sum_{x\in G}\Phi(x)\chi_\xi(-x).
\]

For every additive subspace \(H\le G\) and shift \(a\),

\[
  \mathbb E_{W\sim\operatorname{Unif}(H)}[\Phi(a+W)]
  =\sum_{\xi\in H^\perp}\widehat\Phi(\xi)\chi_\xi(a).
  \tag{11}
\]

The zero mode is exactly the full-uniform reference. Every deviation
from that reference is therefore a nonzero-dual Fourier sum.

\subsection{Standard-coordinate description of the
dual}\label{standard-coordinate-description-of-the-dual}

For

\[
  f(X)=\sum_{t=0}^{255}f_tX^t,
\]

define, for each CRT factor,

\[
  E_\ell(f)=\sum_{j=0}^{127}f_{2j}\alpha_\ell^j,
  \qquad
  O_\ell(f)=\sum_{j=0}^{127}f_{2j+1}\alpha_\ell^j.
  \tag{12}
\]

If \(T=T(r)\) is the inactive CRT set, then

\[
  H=\bigcap_{\ell\in T}\ker E_\ell\cap\ker O_\ell,
\]

and hence

\[
  H^\perp
  =\operatorname{span}\{E_\ell,O_\ell:\ell\in T\}.
  \tag{13}
\]

Equivalently \(H^\perp=\star(\operatorname{Ann}(H))\). This makes the
Fourier correction a deterministic function of the inactive CRT set.

\subsection{One and two inactive
factors}\label{one-and-two-inactive-factors}

For a fixed inactive CRT factor, one parity half of one CBD2 polynomial
is a 128-step linear recurrence whose exact dynamic program gives

\[
  -\log_2 p_{\rm lin}=11.700873155140263\ldots .
\]

The factor is inactive in all three components of \(r\) with probability
\(p_{\rm lin}^6\), giving exponent \(70.205238930841578\ldots\) before
the union over the 128 factors. The source stratum is therefore much too
large to discard as a failure event. Its \textbf{Fourier correction},
however, is tiny because the dual enumerator has only weights 128 and
256.

For two inactive factors, the parity dual code is a two-frequency
exponential code. Its exact weight enumerator depends only on the
multiplicative order of the root ratio. The worst relative Fourier
correction remains far below the scale needed by the final tail. These
one- and two-factor enumerators are reused later in the bivariate
transport.

\subsection{Three inactive factors and certified
anti-concentration}\label{three-inactive-factors-and-certified-anti-concentration}

The remaining question is the probability that one CBD2 polynomial
vanishes on three prescribed CRT factors. For a triple
\(T=\{\ell_1,\ell_2,\ell_3\}\), parity separates and

\[
  P_T=p_T^2,
\]

where

\[
  p_T
  =\Pr\left[\sum_{j=0}^{127} C_j
       (\alpha_1^j,\alpha_2^j,\alpha_3^j)=0\right].
  \tag{14}
\]

The exact CBD2 characteristic function is

\[
  \phi(a)=\cos^4(\pi a/q)\ge0,
\]

so

\[
  p_T=q^{-3}\sum_{\lambda\in\mathbb F_q^3}
       \prod_{j=0}^{127}\phi(\lambda\cdot v_j).
  \tag{15}
\]

The \(\binom{128}{3}=341,376\) factor triples split into 2,667 orbits
under the odd-unit action modulo 256, each of size 128. The orbit counts
according to the generated ratio-order \(d\) are

\[
  1,6,28,120,496,2016
\]

for \(d=4,8,16,32,64,128\) respectively.

To avoid a \(q^3\) Fourier enumeration for every orbit, we use a
32-coordinate Construction-A lattice. For each triple, the first 32
values of the linear form define a \(q\)-ary \([32,3]\) code and an
integer lattice

\[
  L_T=\{z\in\mathbb Z^{32}:z\bmod q\in C_T^{(32)}\}.
\]

Exact Fincke--Pohst enumeration certifies that 2,666 of the 2,667
representatives have no nonzero lattice vector below squared radius
\(3,250,000\). The unique exception is \(T=\{0,16,35\}\), with global
block minimum

\[
  \lambda_1(L_T)^2=2,967,467.
\]

Exactly ten short codewords occur below the common radius.  We spell out the
32-to-128-coordinate step because it is theorem-critical.  Write
\(T=\{\ell_1,\ell_2,\ell_3\}\), set
\(\alpha_i=\zeta^{2\ell_i+1}\), and for
\(\lambda=(\lambda_1,\lambda_2,\lambda_3)\in\mathbb F_q^3\) define
\[
 t_j(\lambda)=\sum_{i=1}^3\lambda_i\alpha_i^j.
\]
For every block \(b\in\{0,1,2,3\}\) and \(0\le j<32\),
\[
 t_{j+32b}(\lambda)
   =\sum_{i=1}^3\bigl(\lambda_i\alpha_i^{32b}\bigr)\alpha_i^j
   =t_j\bigl(\lambda^{(b)}\bigr),
 \qquad
 \lambda_i^{(b)}=\lambda_i\alpha_i^{32b}.
\]
Since every \(\alpha_i\ne0\), the map
\(\lambda\mapsto\lambda^{(b)}\) is a bijection of
\(\mathbb F_q^3\) and preserves nonzeroness.  Thus each of the four
32-coordinate blocks of a nonzero mode is a codeword of the same
\(C_T^{(32)}\), after a bijective change of \(\lambda\).  For every
ordinary orbit, the exact negative certificate therefore gives energy
strictly greater than \(D=3,250,000\) in each block, hence total
energy strictly greater than \(4D=13,000,000\).

For the exceptional orbit \(T=\{0,16,35\}\), if all four blocks exceed
\(D\) the same conclusion holds.  Otherwise shift the low-energy block
to the first position.  The preceding bijection sends it to one of the
ten codewords exhaustively enumerated below radius \(D\).  A shift by
32 positions acts on the 128 values as a cyclic permutation, with the
sign change caused by \(\alpha_i^{128}=-1\) when wrapping; centered
squares are invariant under this signed permutation.  The independent
complete-short-set verification therefore reduces the exceptional case
to those same ten codewords.  Their minimum full-128 centered energy is
\[
  95,302,612>13,000,000.
\]
Consequently every nonzero Fourier mode, including the exceptional
orbit, has total centered energy exceeding \(13,000,000\). Rational
lower bounds on \(\pi\) and upper bounds on \(\ln 2\) then imply
uniformly

\[
  B_T(\lambda)<2^{-33}.
\]

Hence

\[
  P_T
  <
  \frac{31565511365591025}
       {1532436734853303498795276873030959104}
  <2^{-65}.
  \tag{16}
\]

This is the key input for the later \(|T|\ge3\) transport. In that
branch the final affine \textbf{correction} may still be bounded
event-universally by the source-stratum probability (TV at most one),
but the full-ideal failure baseline is retained; the proof therefore
does not replace the entire stratum's failure probability by the
source-stratum probability.

\section{\texorpdfstring{Intermediate benchmark: exact full-ideal \(c_u\) tail}{Intermediate benchmark: exact full-ideal c\_u tail}}\label{exact-c_u-tail}

\subsection{Full-ideal reference
distribution}\label{full-ideal-reference-distribution}

This section is an intermediate benchmark, not a term of the terminal inequality. Barbosa et al. report the familiar simplified-model ML-KEM-768 scale near $2^{-158}$ before the full $c_v$ effect is included \cite{barbosa2025}. The point here is the exact rational tail, agreement of two exact engines, and its diagnostic role in the transport architecture; the numerical value $158.85\ldots$ is not presented as a standalone state-of-the-art breakthrough.

When \(H=R_q\), each coefficient of \(u\) is uniform and the \(d_u=10\)
compression error has exact distribution

\[
  D_U=\frac1{3329}
  \{ -2:128,\;-1:1024,\;0:1024,\;1:1024,\;2:129\}.
  \tag{17}
\]

For the \(c_u\)-only correctness quantity used before reintroducing
\(c_v\), the full-ideal one-coordinate residual can be written as

\[
  N_{\rm ref}=\sum_{j=1}^{768}W_j+e_2,
  \tag{18}
\]

where the \(W_j\) are i.i.d. finite-support variables obtained from

\[
  W=X-Y,
\]

with \(X=ER\) for independent CBD2 \(E,R\) and \(Y=S(E_1+C_U)\) for
independent CBD2 \(S,E_1\) and \(C_U\sim D_U\).

The resulting \(W\) distribution has support \([-12,12]\) and
denominator \(218,169,344\). Its exact histogram is archived in the
artifact.

\subsection{Exact convolution}\label{exact-convolution}

The tail of (18) is computed by two independent exact methods:

\begin{enumerate}
\def\labelenumi{\arabic{enumi}.}
\tightlist
\item
  polynomial exponentiation by Kronecker packing with GMP and a
  carry-separation proof;
\item
  an exact coefficient recurrence for the power \(P_W(x)^{768}\).
\end{enumerate}

Both engines return the same rational number. For one coefficient,

\[
  -\log_2\Pr[|N_{\rm ref}|\ge728]
  =166.8714761099527094221349961610945\ldots,
\]

and after the union bound over 256 coefficients,

\[
  -\log_2 P_{\rm ref,global}
  =158.8714761099527094221349961610945\ldots .
  \tag{19}
\]

Thus the exact reference tail already meets the historical \(2^{-158}\)
\(c_u\) scale.

\subsection{Why scalar Chernoff was
insufficient}\label{why-scalar-chernoff-was-insufficient}

At the best scalar-Chernoff point used in the project, the
full-reference bound has exponent only

\[
  153.7570443700203925\ldots,
\]

losing

\[
  5.1144317399323169\ldots
\]

bits relative to the exact tail. This explains why a correct MGF
factorization alone could not reach the target.

\subsection{Formal Laurent-PGF
transport}\label{formal-laurent-pgf-transport}

Rather than transporting only one MGF evaluation, we retain the formal
coefficient distribution of the rounded term. For a fixed ideal input
\(w\), let

\[
  F_w(X)=\mathbb E_{e,s}[X^{s g_e(w)}]
        =\sum_k c_k(w)X^k.
\]

Equip Laurent polynomials with the weighted coefficient norm

\[
  \left\|\sum_k a_kX^k\right\|_Z
  =\sum_k|a_k|Z^k.
  \tag{20}
\]

This norm is submultiplicative and equals the MGF evaluation on a
genuine probability generating function. For any signed coefficient
discrepancy \(A(X)=\sum_k a_kX^k\) and any coefficient event \(E\), the
elementary inequality

\[
  \left|\sum_{k\in E}a_k\right|\le\sum_k|a_k|
\]

shows why coefficient \(\ell_1\) control dominates the discrepancy of
every tail event; submultiplicativity then propagates that control
through convolution. Parseval applied coefficientwise gives exact
variance bounds for the nonzero Fourier modes. Rational square-root
envelopes yield the uniform certificate

\[
  \delta_{L^1}<\frac34.
  \tag{21}
\]

Combining (21) with the exact one-factor and two-factor dual enumerators
and the three-factor certificate (16) transports the \textbf{tail}, not
merely the MGF, from the full ideal to the actual mixture of ideals. The
resulting absolute affine correction after the 256-coordinate union
bound satisfies

\[
  -\log_2\Delta_{c_u}>165.0337756267520487\ldots .
\]

Therefore

\[
  -\log_2 P_{c_u}
  =158.8514718210732256191632936792095\ldots>158.
  \tag{22}
\]

No probability of a proper ideal is added as a raw failure event.

\section{\texorpdfstring{Exact joint structure of \(c_u\) and
\(c_v\)}{Exact joint structure of c\_u and c\_v}}\label{exact-joint-structure-of-c_u-and-c_v}

\subsection{\texorpdfstring{Uniform-input \(c_v\)
reference}{Uniform-input c\_v reference}}\label{uniform-input-c_v-reference}

For \(U\) uniform in \(\mathbb F_q\), the \(d_v=4\) error

\[
  C_V=R_4(U)-U
\]

has exact distribution

\[
  \Pr[C_V=-104]=\frac8{3329},
\]

\[
  \Pr[C_V=j]=\frac{16}{3329}
  \quad(-103\le j\le103),
\]

and

\[
  \Pr[C_V=104]=\frac9{3329}.
  \tag{23}
\]

Thus \(|c_v|\le104\), recovering the deterministic envelope used by
earlier work, but (23) is only a marginal reference law; it does not
imply independence from \(c_u\).

\subsection{\texorpdfstring{Product ideal governing
\(v\)}{Product ideal governing v}}\label{product-ideal-governing-v}

Let

\[
  H=I(r),\qquad K=I(s).
\]

For fixed \(r,s,e,e_2,\mu\), the matrix-dependent term in \(v\) is

\[
  s^TA^Tr.
\]

Its image as \(A\) varies uniformly is exactly the product ideal

\[
  J=HK=I(r)I(s).
  \tag{24}
\]

Therefore

\[
  v\mid(r,s,e,e_2,\mu)
  \sim
  \operatorname{Unif}(\mu+e^Tr+e_2+J).
  \tag{25}
\]

Whenever \(J\neq0\), every standard coefficient of \(v\) is uniform in
\(\mathbb F_q\), even when \(J\) is proper. Nevertheless, this marginal
uniformity again does not imply independence from \(u\) or \(c_u\).

\subsection{Graph coupling}\label{graph-coupling}

Set

\[
  X=A^Tr,\qquad u=e_1+X,
\]

\[
  h=e^Tr+e_2,
\]

and

\[
  z=s^TR_{10}(u).
\]

The residual before \(c_v\) is

\[
  n'=h+s^TX-z.
  \tag{26}
\]

At the same time,

\[
  v=\mu+h+s^TX=\mu+n'+z.
  \tag{27}
\]

Consequently the final residual is exactly

\[
  \widetilde n
  =n'+c_v
  =R_4(\mu+n'+z)-\mu-z.
  \tag{28}
\]

Equation (28) is the central structural reduction for the joint rounding
problem. It shows why convolving an independent \(D_V\) with the \(c_u\)
residual is not justified: the same matrix-derived quantity \(X\)
determines both \(n'\) and \(z\).

\section{Bivariate reference and affine
transport}\label{bivariate-reference-and-affine-transport}

\subsection{Full-ideal bivariate
primitive}\label{full-ideal-bivariate-primitive}

In the full-ideal model, let \(U\) be uniform in \(\mathbb F_q\) and
\(S,E_1\) independent CBD2 variables. Define

\[
  Q=S(U-E_1),\qquad Z=SR_{10}(U),
\]

or equivalently

\[
  M=Q-Z=-S(E_1+c_u(U)).
  \tag{29}
\]

The exact primitive \((Q,Z)\) has denominator \(852,224\) and 13,197
nonzero states. We use \(M\) rather than \(A\) for the transformed first
coordinate to avoid any collision with the public matrix notation. Two
independent constructions produce identical integer tables. Applying the
exact bijection \((Q,Z)\mapsto(M,Q)=(Q-Z,Q)\) gives the primitive used
for the correctness pair.

For one output coefficient, the full-ideal pair

\[
  (n'\bmod q,\;q_v=v-\mu)
\]

is represented in the group algebra \(\mathbb Q[\mathbb F_q^2]\) as 768
convolutions of the transformed \((M,Q)\) primitive, together with the
diagonal contribution of \(h\).

A useful subtlety is that the marginal of \(v\) is exactly uniform
conditional on \(s\neq0\); the unconditional law retains the exact
branch

\[
  \Pr[s=0]=(3/8)^{768}.
\]

\subsection{\texorpdfstring{Independent-\(c_v\) benchmark and
graph-coupled
reference}{Independent-c\_v benchmark and graph-coupled reference}}\label{independent-c_v-benchmark-and-graph-coupled-reference}

If one \textbf{artificially} replaces \(c_v\) by an independent sample
from \(D_V\), exact convolution gives after the 256-coordinate union
bound

\[
  -\log_2 B^{\rm ind}_0=165.2448187342889453\ldots,
\]

\[
  -\log_2 B^{\rm ind}_1=165.0100815990650095\ldots .
\]

This is only a benchmark.

We now make the spectral-to-total-variation step explicit. Let $\mu(a,b)$ denote the transformed one-primitive probability mass function of $(M,Q)$ on $\mathbb F_q^2$. For $\eta\in\mathbb F_q$, define the partial Fourier slice in the second coordinate by
\[
  \widetilde\mu_\eta(a)=\sum_{b\in\mathbb F_q}\mu(a,b)\exp(-2\pi i\eta b/q),
  \qquad
  \mu(a,b)=\frac1q\sum_{\eta\in\mathbb F_q}\widetilde\mu_\eta(a)\exp(2\pi i\eta b/q),
\]
and set $\rho_\eta=\|\widetilde\mu_\eta\|_1=\sum_a|\widetilde\mu_\eta(a)|$.

\begin{lemma}[Partial Fourier control of the graph coupling]\label{lem:c11b-tv}
For one fixed output coefficient in the full-ideal reference experiment, let $\nu$ be the exact graph-coupled law after the 768 multiplicative contributions and the independent diagonal $h$-convolution, and let $\nu_0$ be the law obtained by retaining the first marginal while replacing the second coordinate by an independent uniform element of $\mathbb F_q$. Then
\[
  \operatorname{TV}(\nu,\nu_0)
  \le \frac12\sum_{\eta\ne0}\rho_\eta^{768}.
\]
\end{lemma}

\begin{proof}
For a fixed standard output coefficient there are $kn=3\cdot256=768$ independent scalar primitive contributions in the full-ideal experiment: the relevant entries of the three secret polynomials and of the independently uniform full-ideal matrix images are independent under the ROM/CBD model. Negacyclic wraparound only changes signs. At the bivariate level those signs do not change the primitive law because CBD2 is symmetric under $S\mapsto-S$ and hence $(M,Q)$ and $(-M,-Q)$ have the same law.

Partial Fourier transform turns convolution in the second coordinate into slice-wise convolution in the first. The zero slice is $\widetilde\mu_0(a)=\sum_b\mu(a,b)$, the first marginal. After 768 convolutions, Fourier inversion therefore identifies the zero-mode contribution exactly with that first marginal tensored with the uniform law on the second coordinate, which is $\nu_0$ before the common diagonal convolution.

Subtracting the zero mode and applying Fourier inversion and the triangle inequality gives
\[
  \|\nu-\nu_0\|_1
   \le \sum_{\eta\ne0}\|\widetilde\mu_\eta^{(*768)}\|_1.
\]
The $\ell_1$ norm is submultiplicative under convolution, so
$\|\widetilde\mu_\eta^{(*768)}\|_1\le\rho_\eta^{768}$. Finally, convolution by the common diagonal law of $h$ is an $\ell_1$ contraction and maps the zero-mode reference to the same first-marginal-times-uniform reference after translation; it cannot enlarge the distance. Dividing the $\ell_1$ distance by two yields the claimed total-variation bound.
\end{proof}

Two independently organized exact rational interval certificates establish
\[
  \rho_{\pm1024}<\frac{853}{1000},
  \qquad
  \rho_\eta<\frac35
  \quad\text{for the other 3326 nonzero modes}.
\]
For scale only, the primary certificate reports the non-decisional upper diagnostics $\rho_{1024}<0.852925358017418$ and a largest other-mode upper diagnostic $\rho_{1281}<0.595483743750827$; the rational inequalities are the decision values. Lemma~\ref{lem:c11b-tv} therefore gives
\[
  \operatorname{TV}_{\rm graph}
  \le\frac12\left[2\left(\frac{853}{1000}\right)^{768}
                 +3326\left(\frac35\right)^{768}\right],
  \tag{30}
\]
with $-\log_2\operatorname{TV}_{\rm graph}>176.1656473680\ldots$ as a non-decisional diagnostic.

The resulting \textbf{full-ideal graph-coupled} bounds are

\[
  -\log_2 B^{\rm graph}_0>165.0658792859162109\ldots,
\]

\[
  -\log_2 B^{\rm graph}_1>164.8566361339541356\ldots .
  \tag{31}
\]

Let \(p_b\) denote the true one-coordinate failure probability in the
full-ideal graph-coupled model, let \(b_b^{\rm graph}\) be the certified
one-coordinate upper bound, and define
\(B_b^{\rm graph}=256b_b^{\rm graph}\). Equation (31) establishes the
ordering of the \emph{certified bounds},
\(b_0^{\rm graph}<b_1^{\rm graph}\), with
\(p_b\le b_b^{\rm graph}\) for each bit. It does not assert or require any ordering between the two true bit-specific probabilities.

\subsection{Bivariate affine
observable}\label{bivariate-affine-observable}

To transport (31) from the full ideal to \(H=I(r)\), define for one
input coefficient \(w\)

\[
  G_w(M,Q)
  =\operatorname{Law}_{E_1,S}
   \bigl(M=-S(E_1+c_u(E_1+w)),\;Q=Sw\bigr).
  \tag{32}
\]

Equivalently,

\[
  M=S(w-R_{10}(E_1+w)),\qquad Q=Sw.
\]

For a fixed output coefficient, the 768 multiplicative terms are formed
from three independent copies of the 256-coordinate \(H\)-average of
(32). The remaining \(h=e^Tr+e_2\) is a diagonal translation in the
output group; convolution by that law is an \(\ell_1\) contraction. Thus
no new joint census of \((I(r),h)\) is required.

\subsection{Bivariate Fourier identity and spectral
bound}\label{bivariate-fourier-identity-and-spectral-bound}

For normalized input Fourier coefficients

\[
  \widehat G_\xi
   =q^{-1}\sum_{w\in\mathbb F_q}\chi_\xi(-w)G_w,
\]

\(G_w\) is a probability measure on the \textbf{output} group
\(\mathbb F_q^2\) with coordinates \((M,Q)\), whereas the Fourier
transform above is taken over the \textbf{input} scalar
\(w\in\mathbb F_q\). Hence \(\widehat G_\xi\) is generally a
signed/complex group-algebra element on \(\mathbb F_q^2\); products over
standard input coordinates become convolutions in that output group. The
normalization is the same normalized additive convention as in (11).

The annihilator identity gives, for one polynomial component,

\[
  K_H
  =\sum_{\xi\in H^\perp}
    \mathop{*}_{b=0}^{255}\widehat G_{\xi_b}.
  \tag{33}
\]

For \(H=R_q\), only the zero mode remains.

For every nonzero input frequency, the \(S=0\) contribution vanishes by
character orthogonality. Pairing \((S,E_1,w)\) with \((-S,-E_1,-w)\)
reduces the remaining norm to an average of \(|\cos|\), except at four
exact rounding-boundary residues. Counting those defects exactly and
using the identity for the discrete \(|\cos|\) average gives

\[
  \|\widehat G_\xi\|_1
  \le
  \frac{5}{8q\sin(\pi/(2q))}+\frac{5}{2q}.
\]

With the rational inequalities \(\pi>333/106\) and \(\sin x>x-x^3/6\),

\[
  \|\widehat G_\xi\|_1
  <
  \frac{176091287861935481}{441720252460831626}
  <\frac25.
  \tag{34}
\]

An independent audit proves the same \(2/5\) threshold using the weaker
\(\pi>3.14\).

\subsection{Proper-ideal correction}\label{proper-ideal-correction}

Set \(\delta=2/5\). For one inactive CRT factor, the exact dual
enumerator is

\[
  W_1(x)=1+2(q-1)x^{128}+(q-1)^2x^{256}.
\]

Combined with the exact one-factor source probability, its weighted
affine correction has exponent exceeding \(219.12\) bits.

For two inactive factors, the worst parity enumerator is

\[
  W_{\rm par}(x)
   =1+2(q-1)x^{64}+(q-1)^2x^{128},
\]

with full enumerator \(W_2=W_{\rm par}^2\). The two-factor source
probability is bounded by reusing the three-factor certificate with the
third Fourier coefficient set to zero. The weighted contribution exceeds
\(197.72\) bits.

For at least three inactive factors, no spectral enumerator is needed:
(16), cubed across the three independent polynomials of \(r\) and
union-bounded over \(\binom{128}{3}\) triples, gives

\[
  -\log_2P_{\ge3,r\neq0}>177.807118758940241\ldots .
\]

The exact branch

\[
  r=0
\]

is retained separately with probability \((3/8)^{768}\). The branch
\(s=0\) has zero affine correction before the common diagonal
translation. Cases with \(r,s\neq0\) but \(I(r)I(s)=0\) are included in
the observable (32); no raw \(\Pr[I(r)I(s)=0]\) is added as a failure
event.

The logical roles of the strata are summarized below.

\begin{longtable}[]{@{}
  >{\raggedright\arraybackslash}p{(\columnwidth - 6\tabcolsep) * \real{0.2500}}
  >{\raggedright\arraybackslash}p{(\columnwidth - 6\tabcolsep) * \real{0.2500}}
  >{\raggedright\arraybackslash}p{(\columnwidth - 6\tabcolsep) * \real{0.2500}}
  >{\raggedright\arraybackslash}p{(\columnwidth - 6\tabcolsep) * \real{0.2500}}@{}}
\toprule\noalign{}
\begin{minipage}[b]{\linewidth}\raggedright
Inactive-factor branch
\end{minipage} & \begin{minipage}[b]{\linewidth}\raggedright
Source control
\end{minipage} & \begin{minipage}[b]{\linewidth}\raggedright
Transport control
\end{minipage} & \begin{minipage}[b]{\linewidth}\raggedright
Absolute correction scale
\end{minipage} \\
\midrule\noalign{}
\endhead
\bottomrule\noalign{}
\endlastfoot
\(|T|=1\) & exact one-factor probability & exact dual enumerator at
\(\delta=2/5\) & exponent \(>219.12\) \\
\(|T|=2\) & exact/certified two-factor probability & ratio-order dual
enumerator & exponent \(>197.72\) \\
\(|T|\ge3\), \(r\ne0\) & three-factor certificate + union bound &
event-universal TV \(\le1\) & exponent \(>177.807\) before global
combination \\
\(r=0\) & exact \((3/8)^{768}\) & separate branch & negligible at
theorem scale \\
\(s=0\) & exact branch & affine correction exactly zero before common
translation & zero correction \\
\end{longtable}

Thus the coarse \(|T|\ge3\) step bounds only the \textbf{transport
discrepancy} by its source probability; the full-ideal event probability
is still present in the baseline.

Combining all strata and then applying the 256-coordinate union bound
gives the event-universal affine correction

\[
  -\log_2\Delta_{\rm affine,global}
  =169.8071173061514638884937022909900\ldots .
  \tag{35}
\]

The recomputed branch diagnostics are $219.1266328576\ldots$ bits for $|T|=1$, $197.7286285332\ldots$ bits for $|T|=2$, $177.8071187589\ldots$ bits for $|T|\ge3$ with $r\ne0$, and $1086.7487994462\ldots$ bits for $r=0$. Thus the $|T|\ge3$ branch dominates the affine correction. These values are derived from the exact frozen fractions and are diagnostics rather than threshold decisions.

\section{Terminal bit-specific bound}\label{terminal-bit-specific-bound}

The broad symmetric surrogate \(|\widetilde n|\ge832\) is unnecessary
for the final theorem. For an arbitrary message
\(m=(m_1,\ldots,m_{256})\), the full-ideal message-failure probability
obeys the union bound

\[
\begin{aligned}
  \Pr[\text{message failure in the full-ideal model}]
  &\le \sum_{i=1}^{256}p_{m_i}\\
  &\le \sum_{i=1}^{256}b^{\rm graph}_{m_i}\\
  &\le 256b^{\rm graph}_1
   = B^{\rm graph}_1.
\end{aligned}
  \tag{36}
\]

The last inequality uses only the certified ordering
\(b^{\rm graph}_0\le b^{\rm graph}_1\) from C11B/C12B. No ordering of
the true probabilities \(p_0,p_1\) is assumed.

The correction (35) is event-universal, so the same global correction
may be added to the certified bit-1 graph-coupled bound:

\[
  P_*
  =B^{\rm graph}_{1,\,\text{global}}
   +\Delta_{\rm affine,global}.
  \tag{37}
\]

Let \(N_*/D_*\) denote the reduced numerator and denominator of (37),
stored in the artifact in canonical text form \texttt{N/D}. The
canonical digest is

\begin{Shaded}
\begin{Highlighting}[]
\scriptsize\NormalTok{46bcd82e73b61b2d0c37fc894a311fc1adf67f99ed756b9989866780fd51ea47}\normalsize
\end{Highlighting}
\end{Shaded}

and

\[
  -\log_2 P_*
  =164.810716201343121023033838332580034777995225649259958984\ldots .
  \tag{38}
\]

The threshold claims are checked without floating point. For
\(\lambda=164.8=824/5\),

\[
  N_*^5 2^{824}\le D_*^5,
  \tag{39}
\]

and for \(\lambda=164.81=16481/100\),

\[
  N_*^{100}2^{16481}\le D_*^{100}.
  \tag{40}
\]

The stronger \(164.82\) threshold fails, as expected for the present
certificate. The rigorous sensitivity calculation brackets the excess exponent above $164.81$ by $0.0007162013431210230\ldots$ bit. At the same target, the certified global affine correction has about $1.585\%$ relative slack; expressed at the three-factor source parameter $u_3$, the corresponding relative slack is about $0.526\%$. These diagnostics do not weaken the exact comparison, but they show that $164.81$ is a tight certificate rather than a perturbatively robust round number.

\[
  \Pr[K'\neq K]\le P_*\le2^{-164.81}.
\]

This is an upper bound obtained using exact finite computations,
rational inequalities, total-variation domination, and a union bound
over 256 coefficients; it is not an exact DFR.

\paragraph{Proof of Theorem~\ref{thm:main}.}
The exact bit-specific FIPS regression gives the safe regions $[-832,832]$ for bit 0 and $[-832,831]$ for bit 1. Lemma~\ref{lem:c11b-tv} and the rational C11B spectral certificates give $p_b\le b_b^{\rm graph}$ and the certified ordering $b_0^{\rm graph}\le b_1^{\rm graph}$. The bivariate affine identity and the $|T|=1,2,\ge3$, $r=0$, $s=0$, and $I(r)I(s)=0$ handling give the event-universal correction in (35), whose theorem-critical $|T|\ge3$ source bound is replayed by C10C. Hence (36)--(37) and the union bound over 256 coefficients give $\Pr[\text{K-PKE message failure}]\le P_*$. Honest K-PKE message recovery deterministically implies honest ML-KEM key recovery, so $\Pr[K'\ne K]\le P_*$. Finally the exact integer comparisons (39)--(40) certify $P_*\le2^{-164.81}$. The equal-strength secondary C11C assembly with independently derived theorem-dominant upstream inputs reconstructs the same terminal digest without reading the primary C11C output.

\section{Methods: computational and AI
systems}\label{methods-computational-and-ai-systems}

We disclose the computational workflow used to discover, challenge, and verify candidate arguments, and distinguish model-assisted exploration from theorem-level evidence. Different model families were used in separate exploratory and adversarial roles, interleaved with local exact computation, executable checks, checkpointing, and explicit cross-checks. We describe this as a Networked Heterogeneous Systems for Artificial Intelligence approach. It is an engineering workflow, not a cryptographic contribution or a formally validated scientific method.

\subsection{Author expertise and human
verification}\label{author-expertise-and-human-verification}

Both authors are engineers by training and work on the architecture of
AI systems. C.T.'s professional field is cybersecurity.
Neither author holds an academic appointment in lattice cryptography or formal methods. C.T. read the mathematical chain, phase
reports, counter-audits, and final manuscript step by step; A.D.
directed and ran the computational workflow and reproducibility
pipeline. Both authors selected which candidate arguments to retain and
are jointly responsible for the statements made here. This disclosure is
intended to make the verification boundary explicit, not to replace
specialist peer review. At the time this package was prepared, the
manuscript had not yet received specialist human peer review.

\subsection{Model-assisted exploration and adversarial
review}\label{model-assisted-exploration-and-adversarial-review}

OpenAI ChatGPT (GPT-5.6 Sol) was used for literature synthesis,
exploratory derivations, code generation and debugging, adversarial
review of candidate arguments, checkpoint construction, and manuscript
editing. Qwen3.5-9B was run locally as a secondary exploration engine. MathLab served as the MCP orchestrator coordinating model calls, exact computation, checkpointing, and reproducibility tasks. During post-exploration audit stages, Anthropic Claude (Opus 5)
and Mistral Medium 3.5 were used as distinct-family cross-checks for
counter-analysis, referee-style review, and audits of selected
derivations and computational outputs.

Iterative review cycles were used for defect discovery. An issue was considered closed only after human review and/or executable mathematical evidence, or was recorded as an open limitation; model agreement was not a proof criterion. These machine-assisted checks are not mathematical proof, specialist peer review, or statistically independent evidence, because model families can share data and failure modes. The theorem-level evidence is the explicit mathematical chain and the exact artifacts described below.

\subsection{Exact computation and proof
discipline}\label{exact-computation-and-proof-discipline}

The V8 release environment used Python 3.13.5 for arbitrary-precision integers, \texttt{fractions.Fraction}, parsing, and modular arithmetic; GCC/g++ 14.2.0 in C++17 mode for exhaustive finite engines; GNU MP 6.3.0 for the exact C10E engines; SymPy 1.14.0 where symbolic exact algebra is actually used; and mpmath 1.3.0 only for decimal diagnostics after exact quantities have been formed. Rational interval arithmetic is used for spectral decisions. Shell/build tools orchestrate the replay, and pdfTeX 1.40.26 (TeX Live 2025/dev) builds the manuscript. The release logs report observed local wall times rather than universal performance claims.

Every theorem-level accept/reject decision in the final chain is based
on integer, modular, rational, exact symbolic, or rigorously outward
rational-interval arithmetic. Floating-point values are presentation or
non-decisional diagnostics only. Exhaustive computations are used only
where the finite search space has first been reduced to a canonical,
explicitly enumerated family. Sampling is never used as proof. Critical
computational layers have an independently organized implementation or
counter-audit, and the terminal fraction has a third reconstruction path.

\section{Verification and
reproducibility}\label{verification-and-reproducibility}

All publication-critical code, compact data, certificates, and expected
results referenced by this manuscript are contained in the
\texttt{anc/} directory distributed with the source. No theorem claim in
this manuscript depends on a file outside that directory or on network
access.

\subsection{Verification levels and entry points}\label{reproducibility-entry-points}

Every verifier mode is read-only with respect to the frozen source tree. The driver disables Python bytecode generation, hashes the complete ancillary tree before and after each mode, and executes historical mutating producers only inside temporary copies. A successful mode therefore ends with \texttt{PASS\_FROZEN\_TREE\_IMMUTABILITY\_V8}.

From the root of a fresh arXiv-source extraction, the principal commands are
\begin{Shaded}
\begin{Highlighting}[]
\ExtensionTok{python3}\NormalTok{ anc/verify\_publication.py --quick}
\ExtensionTok{python3}\NormalTok{ anc/verify\_publication.py --full}
\ExtensionTok{python3}\NormalTok{ anc/verify\_publication.py --rebuild-c10c}
\ExtensionTok{python3}\NormalTok{ anc/verify\_publication.py --rebuild-c10e}
\end{Highlighting}
\end{Shaded}

\cliopt{quick} checks integrity, theorem-critical paths, the absence of optimization-sensitive assertions and machine-specific paths, the bit-specific regression, the canonical terminal digest, and the exact $164.8/164.81/164.82$ comparisons. \cliopt{full} recomputes the short/medium routes in a temporary work tree: C11B primary/counter, C11C primary, the V8 secondary route using the independently derived C10C value of $u_3$ and a fresh C11B counter result, the deliberately weaker sanity route, C12B, C13A, primary/counter sensitivity brackets, proof lints, and a fast in-memory corruption smoke. It does not run the long C10C search. \cliopt{red-team} runs the full six-mutation filesystem/data regression in both normal and optimized Python modes (12 rejection checks). \cliopt{rebuild-c10c} regenerates the orbit representatives, runs the primary exhaustive engine and a genuinely separate V8 exact interval-enumeration verifier on all 2,667 orbits, verifies the complete exceptional short set independently, and performs the rational derivation of $P_T$; it supports sharding/resume outside the frozen tree. \cliopt{rebuild-c10e} replays the two exact non-terminal benchmark engines.

\subsection{Artifact inventory}\label{artifact-inventory}

The detailed claim-to-artifact map is
\texttt{anc/ARTIFACT\_INDEX\_V8.md}. Principal entry points are:

\begin{longtable}[]{@{}p{0.43\columnwidth}p{0.53\columnwidth}@{}}
\toprule
Artifact & Role \\
\midrule
\endhead
\path{anc/verify_publication.py} & read-only V8 verifier and orchestrator \\
\path{anc/code/c11b_primary.py} / \path{anc/code/c11b_counter_audit.py} & C11B primary and independently organized counter routes, including distinct rational spectral certificates \\
\path{anc/results/r1dc11b/primary_result.json} / \path{anc/results/c11b_counter_v8.json} & canonical primary and current counter C11B outputs \\
\path{anc/code/r1dc10c_svp_exact_R3250000.cpp} & primary exact C10C Fincke--Pohst engine \\
\path{anc/code/c10c_independent_exact_verify_v8.cpp} & independently implemented exact V8 interval-enumeration verifier for all 2,667 C10C orbits and the exceptional complete short set \\
\path{anc/code/c10c_full_search_v8.py} & fresh orbit reconstruction plus sharded/resumable primary and independent C10C replay \\
\path{anc/code/r1dc10c_counter_audit.py} & direct representative binding and independent rational derivation of $P_T$ from V8 search/refinement evidence \\
\path{anc/code/r1dc10e_tail_pack.cpp} / \path{anc/code/r1dc10e_tail_recurrence.cpp} & exact non-terminal full-ideal $c_u$ benchmark engines \\
\path{anc/code/r1dc11c_primary.py} & primary bivariate affine transport \\
\path{anc/code/r1dc11c_equal_strength_counter_audit.py} & equal-strength V8 secondary assembly using the independently derived C10C $P_T$ and C11B-counter graph bounds \\
\path{anc/results/r1dc11c/secondary_counter_v8.json} & frozen V8 secondary-route result and scoped independence contract \\
\path{anc/code/r1dc12b_*.py} / \path{anc/code/r1dc13a_third_terminal_audit.py} & distinct terminal assemblies with explicitly documented shared certified upstreams \\
\path{anc/code/terminal_sensitivity_v8.py} & primary/counter exact sensitivity certificates \\
\path{anc/results/r1dc11c/terminal_sensitivity_primary_v8.json} / \path{anc/results/r1dc11c/terminal_sensitivity_counter_v8.json} & frozen V8 primary/counter sensitivity outputs with exact input hashes \\
\path{anc/results/terminal_fraction.json} & canonical terminal numerator/denominator and digest \\
\path{anc/ARTIFACT_INDEX_V8.md} / \path{anc/ALL_ARTIFACTS_INDEX_V8.json} / \path{anc/BUILD_MATRIX_V8.json} & claim map, exhaustive capsule inventory, and executable build/smoke contract \\
\path{anc/ENVIRONMENT.md} / \path{anc/SHA256SUMS.txt} & environment contract and ancillary integrity manifest \\
\bottomrule
\end{longtable}

\subsection{What is computationally
exhaustive}\label{what-is-computationally-exhaustive}

The theorem-critical exhaustive component is the 341,376-triple/2,667-orbit C10C Construction-A analysis. The primary Fincke--Pohst engine and a separately implemented exact V8 interval-enumeration verifier establish the finite negative-search facts; the independent complete-short-set verification establishes the exceptional minimum and full-128 energy; and a separate rational counter-audit converts those facts into the certified $P_T$ bound used downstream. The compact C10C checker is only a frozen-output binding/consistency check. C11B uses exact bivariate primitive tables and rational spectral certificates, while C11C/C12B/C13A use exact rational/integer reconstruction. The retained C10E benchmark uses two exact polynomial-tail engines but is non-terminal. Historical low-weight development material is excluded from the V8 publication capsule because it is not a theorem dependency or publication claim.

\section{Limitations and discussion}\label{limitations-and-discussion}

\subsection{ROM/CBD versus fixed SHAKE}\label{romcbd-versus-fixed-shake}

The largest conceptual limitation is the model boundary. The public
matrix and small-noise distributions are treated through the explicit
ROM/CBD abstraction of the theorem. This does not establish an
information-theoretic equality with the literal fixed-SHAKE
implementation of FIPS 203, nor does it separately account for every
possible equality/collision among finite seeds or PRF/XOF inputs in that
implementation. The result should therefore be cited as a certified
correctness bound \textbf{in the stated abstraction}.

\subsection{Upper bound versus exact
DFR}\label{upper-bound-versus-exact-dfr}

The final quantity \(P_*\) is not an exact probability. It contains
several one-sided steps: Fourier norm bounds, total-variation
domination, stratum union bounds, and finally a union bound over the 256
decoded message coefficients. None of these requires output-coordinate
independence, but they do make the result a certified upper bound.

\subsection{\texorpdfstring{Relation to FIPS's reported
\(2^{-164.8}\)}{Relation to FIPS's reported 2\^{}\{-164.8\}}}\label{relation-to-fipss-reported-2-164.8}

The numerical proximity between (38) and the \(2^{-164.8}\) value
reported by FIPS is useful context, but the logical status is different.
We do not infer our theorem from the FIPS estimate; instead, our
exact/rational chain produces a bound whose exponent happens to reach
that scale under the ROM/CBD assumptions.

\subsection{No new IND-CCA claim}\label{no-new-ind-cca-claim}

The article proves no new IND-CCA theorem and no new Fujisaki--Okamoto
reduction. Machine-checked cryptographic reductions are available
elsewhere \cite{almeida2024}. We bound a correctness event and use only
the deterministic containment (3) for honest encapsulations.

\subsection{Fixed-message versus adaptive correctness}\label{fixed-vs-adaptive}
The probability experiment in Theorem~\ref{thm:main} fixes an arbitrary message independently of the public matrix and all key/noise randomness. Some Fujisaki--Okamoto analyses use a stronger adaptive $\delta$-correctness notion in which an adversary may choose a message after observing public data. This work does not prove an automatic passage from the fixed-independent-message experiment to that adaptive notion.

\subsection{Tight threshold and proof boundary}\label{tight-threshold}
The $164.81$ certificate is exact but numerically tight: the terminal exponent has only about $0.0007162$ bit of margin and $164.82$ fails. The full human mathematical chain is not machine-checked in a proof assistant; executable checks certify the finite and rational subclaims described in the artifact map but do not replace specialist proof review. No automatic generalization to ML-KEM-512 or ML-KEM-1024 is claimed. The retired V2 low-weight standalone result is not needed by the terminal proof.

\subsection{Standard maintenance}\label{standard-maintenance}

The NIST publication page for FIPS 203 contains a planning note that an
issue will be corrected in a future revision \cite{fips203}. The
theorem is keyed to the audited algorithms, compression maps, and
decoding conventions. Any revision affecting those details requires a
targeted normative re-audit.

\subsection{Future research}\label{future-research}

Several extensions appear worthwhile, and none is needed for the theorem
proved here.

\begin{enumerate}
\def\labelenumi{\arabic{enumi}.}
\tightlist
\item
  \textbf{Bridge the ROM/CBD model to the fixed FIPS instantiation.} The
  cleanest conceptual next step is to quantify, rather than merely
  state, the gap between independent random-function streams/CBD
  primitives and the finite-seed SHAKE/PRF implementation. Any such
  result would have to account explicitly for domain separation, seed
  collisions, and the exact sampling interface.
\item
  \textbf{Machine-check the algebraic transport.} The finite
  certificates are independently checkable, but the full proof chain is
  not currently formalized in a proof assistant. Encoding the
  ideal/coset reduction, Fourier identities, event containments, and
  rational inequalities in EasyCrypt, Lean, Coq, or a comparable system
  would materially strengthen assurance.
\item
  \textbf{Extend the failure-aware transport to ML-KEM-512 and
  ML-KEM-1024.} The present paper deliberately treats only ML-KEM-768.
  The same architecture may transfer, but the low-weight and tail
  constants must be recomputed rather than assumed.
\item
  \textbf{Generalize beyond ML-KEM.} The combination of proper-ideal
  stratification with event-specific multivariate transport may be
  useful for other Ring/Module-LWE constructions where rounding noise is
  generated from algebraically coupled quantities.
\item
  \textbf{Reduce remaining one-sided slack.} The present result is an
  upper bound. Event-specific spectral information may permit a closer
  characterization of the actual ROM/CBD DFR, but such a refinement
  should remain secondary to preserving the dependency structure and
  proof auditability.
\end{enumerate}

\section{Conclusion}\label{conclusion}

We have given a dependency-preserving correctness analysis of ML-KEM-768 in an explicit ROM/CBD abstraction. The terminal argument combines a graph-coupled full-ideal bivariate reference, an explicit partial-Fourier/total-variation proof, proper-ideal affine transport with an exhaustively replayable three-factor rare-stratum bound, and exact bit-specific FIPS events. The low-weight V2 analysis and the exact full-ideal $c_u$ benchmark are not terminal dependencies. The final rational certificate satisfies

\[
  \Pr[K'\neq K]\le2^{-164.81}
\]

for honest decapsulation in the stated model. The equal-strength independent C11C route closes the same threshold, while the older weaker counter-route is retained only as a sanity check.

The most important qualification is also the simplest: this is not an
exact DFR and not an information-theoretic statement about the fixed
SHAKE instantiation. Within its declared model, however, the result
brings a dependency-preserving certified upper bound to the same
numerical scale as the heuristic failure value reported for ML-KEM-768,
while preserving the joint rounding structure rather than replacing it
by independent compression noise.

\section{Acknowledgements}\label{acknowledgements}

We are grateful to the authors of the works closest to ours for making
their analyses and, where available, reproducibility material public. In
particular, the papers of Jan-Pieter D'Anvers, Frederik Vercauteren, and
Ingrid Verbauwhede, and of Boyue Fang, Weize Wang, and Yunlei Zhao, made
the role of dependency effects in lattice-scheme failures especially
clear. The formal-verification work of José Bacelar Almeida and
coauthors, and of Manuel Barbosa, Matthias J. Kannwischer, Thing-han
Lim, Peter Schwabe, and Pierre-Yves Strub, provided essential reference
points for separating formal correctness statements from sharp concrete
failure estimates. Abdul Qayyum and Sergey V. Bezzateev's public
artifact was particularly valuable in clarifying which ideal/CRT/Fourier
ingredients are already established prior art and where our own
refinement begins. This acknowledgement does not imply that any of these
authors reviewed, endorsed, or is responsible for the present
manuscript.

\textbf{Generative AI disclosure.} OpenAI ChatGPT (GPT-5.6 Sol) was used for literature synthesis, exploratory derivations, code generation/debugging, adversarial review, reproducibility work, and manuscript editing. Qwen3.5-9B, Anthropic Claude (Opus 5), and Mistral Medium 3.5 were used in the roles described in Methods. All retained text, references, mathematical claims, and computational conclusions were reviewed by the authors; theorem-level numerical claims were accepted only through the explicit mathematical and executable verification chain. The authors take full responsibility for the manuscript.

\textbf{Funding.} This work received no external grant funding. Research costs were borne internally by netHsys SARL and Tommasini Conseil.

\textbf{Competing interests.} The authors declare no competing interests.

\section{Data and code availability}\label{data-and-code-availability}

All code, certificates, and compact data required to verify the
publication-critical computational claims are included in the
\texttt{anc/} verification capsule distributed with this manuscript. The
capsule contains deterministic verification commands and exact output
digests. Historical computations that do not support a V8 publication claim are
excluded from the submission capsule. No external proprietary dataset and no
network access are required for the theorem-critical proof checks. The same verification material is included with the arXiv source as ancillary
material in the \texttt{anc/} directory.

\section{Author contributions}\label{author-contributions}

Both authors defined the research programme, directed the exploration,
and decided which candidate results to retain. C.T. read and checked the
mathematical derivations, phase reports, and proof chain step by step,
participated in the construction of independent verification paths, and
prepared the manuscript. A.D. designed and ran the computational
workflow, managed the exact enumerations and checkpoint pipeline, and
participated in verification and manuscript preparation. Both authors
approved the submitted version and are jointly responsible for the
scientific claims and for the disclosure of computational and AI-system
use.

\section{Appendix A. Three-factor Construction-A
certificate}\label{appendix-a.-three-factor-construction-a-certificate}

For a CRT triple \(T\), the parity-half Fourier sequence is

\[
  t_j=\lambda_1\alpha_1^j+\lambda_2\alpha_2^j+\lambda_3\alpha_3^j.
\]

A block of 32 values defines a \(q\)-ary code of dimension three. In
systematic form \([I_3\mid A_T]\), the corresponding Construction-A
lattice has determinant \(q^{29}\). The primary Fincke--Pohst search and the separately implemented V8 exact interval-enumeration verifier certify the common radius for 2,666 orbits.
For the unique exception \(T=\{0,16,35\}\), the complete short set is
enumerated and its full 128-coordinate energy is computed explicitly.
The primary exceptional refinement and the independent V8 complete-short-set verifier reproduce the same ten short codewords and the same global block minimum \(2,967,467\).

\section{Appendix B. Exact rounding
distributions}\label{appendix-b.-exact-rounding-distributions}

\subsection{\texorpdfstring{\(d_u=10\)}{d\_u=10}}\label{d_u10}

For uniform \(U\in\mathbb F_{3329}\),

\[
  \Pr[c_u=-2]=128/3329,
\]

\[
  \Pr[c_u=-1]=\Pr[c_u=0]=\Pr[c_u=1]=1024/3329,
\]

\[
  \Pr[c_u=2]=129/3329.
\]

\subsection{\texorpdfstring{\(d_v=4\)}{d\_v=4}}\label{d_v4}

The exact distribution is given in (23). In particular
\(\max|c_v|=104\), but the graph-coupled analysis never upgrades the
marginal law to an independence assumption.

\section{Appendix D. Terminal
certificate}\label{appendix-d.-terminal-certificate}

The canonical terminal rational number is stored at
\path{anc/results/terminal_fraction.json}. Its SHA-256 digest for the
canonical text representation \texttt{N/D} is

\begin{Shaded}
\begin{Highlighting}[]
\scriptsize\NormalTok{46bcd82e73b61b2d0c37fc894a311fc1adf67f99ed756b9989866780fd51ea47}\normalsize
\end{Highlighting}
\end{Shaded}

The scripts \texttt{anc/code/r1dc12b\_primary.py},
\texttt{anc/code/r1dc12b\_counter\_audit.py}, and
\texttt{anc/code/r1dc13a\_third\_terminal\_audit.py} independently
reconstruct the terminal quantities used in (37)--(40). The decimal in
(38) is never used as a decision value.

\section{Appendix E. Artifact map}\label{appendix-e.-artifact-map}

The complete publication-facing claim map is shipped as
\texttt{anc/ARTIFACT\_INDEX\_V8.md}. Every path cited there is relative to
the source root and remains inside \texttt{anc/}; the verification
capsule has no proof dependency on project-history material, audit
prompts, or an external release archive.

  \bibliography{refs}

\end{document}